\documentclass[lettersize,journal]{IEEEtran}
\usepackage{amsthm,amssymb}
\newtheorem{theorem}{Theorem}
\usepackage{amsmath,amsfonts}
\usepackage{algorithmic}
\usepackage{algorithm}
\usepackage{array}
\usepackage[caption=false,font=normalsize,labelfont=sf,textfont=sf]{subfig}
\usepackage{textcomp}
\usepackage{stfloats}
\usepackage{url}
\usepackage{verbatim}
\usepackage{graphicx}
\usepackage{cite}
\usepackage{booktabs}
\usepackage{multirow}
\begin{document}
\title{Robust Indoor Localization via Conformal Methods and Variational Bayesian Adaptive Filtering}

\author{Zhiyi Zhou, Dongzhuo Liu*, Songtao Guo,~\IEEEmembership{Senior Member,~IEEE,} and Yuanyuan Yang,~\IEEEmembership{Fellow,~IEEE}
\thanks{Preprint.}
}



\maketitle
\begin{abstract}
Indoor localization is critical for IoT applications, yet challenges such as non-Gaussian noise, environmental interference, and measurement outliers hinder the robustness of traditional methods. Existing approaches, including Kalman filtering and its variants, often rely on Gaussian assumptions or static thresholds, limiting adaptability in dynamic environments. This paper proposes a hierarchical robust framework integrating Variational Bayesian (VB) parameter learning, Huber M-estimation, and Conformal Outlier Detection (COD) to address these limitations. First, VB inference jointly estimates state and noise parameters, adapting to time-varying uncertainties. Second, Huber-based robust filtering suppresses mild outliers while preserving Gaussian efficiency. Third, COD provides statistical guarantees for outlier detection via dynamically calibrated thresholds, ensuring a user-controlled false alarm rate. Theoretically, we prove the Semi-positive Definiteness of Huber-based Kalman filtering covariance and the coverage of sliding window conformal prediction. Experiments on geomagnetic fingerprint datasets demonstrate significant improvements: fingerprint matching accuracy increases from 81.25\% to 93.75\%, and positioning errors decrease from 0.62–6.87 m to 0.03–0.35 m. Comparative studies further validate the framework’s robustness, showing consistent performance gains under non-Gaussian noise and outlier conditions.
\end{abstract}


\section{Introduction}
\IEEEPARstart{I}{ndoor} localization plays a pivotal role in enabling a wide range of applications, from autonomous robotics and wearable devices to smart buildings within the Internet of Things (IoT) ecosystem. As a foundational technology, accurate and reliable indoor localization is essential for applications such as navigation, asset tracking, and environmental monitoring. However, the inherent challenges of indoor environments—such as severe signal degradation caused by multipath effects and dynamic interference—complicate the design of robust localization systems \cite{zhu2024enabling}, \cite{xu2019indoor}. These factors introduce substantial measurement noise, making it difficult to ensure both the accuracy and reliability of indoor positioning.

Despite significant progress, current indoor localization methods are hindered by three major limitations:
\begin{itemize}
    \item Reliance on simplifying Gaussian noise assumptions, which do not hold in many practical scenarios.
    \item Fixed robustness thresholds that lack adaptability in dynamic environments.
    \item Absence of rigorous uncertainty quantification for robust decision-making.
\end{itemize}

Traditional filtering approaches, such as the Kalman Filter (KF) and its nonlinear variants, are optimal only under restrictive assumptions of Gaussian noise and linear system dynamics \cite{urrea2021kalman}. Although the Unscented Kalman Filter (UKF) improves handling of nonlinearities through the unscented transform, it remains vulnerable to non-Gaussian noise and outliers. Adaptive methods, such as variational Bayesian (VB) filters \cite{lin2021variational} and Huber M-estimation \cite{fauss2021minimax}, partially mitigate these issues by addressing time-varying noise and suppressing moderate outliers. However, VB filters introduce high computational overhead, and Huber M-estimation's fixed threshold is not robust to dynamically changing environments. While multi-sensor fusion can improve accuracy, it often inherits distributional assumptions and model errors \cite{erhan2021smart}.

To address these limitations, we propose a novel hierarchical framework that combines variational Bayesian learning, Huber robust filtering, and Conformal Outlier Detection (COD). This framework, built on the principles of Conformal Prediction (CP), offers a robust solution to the challenges of adaptability and reliability in indoor positioning tasks. By leveraging conformal inference, our approach provides statistically rigorous confidence intervals to detect and mitigate the impact of outliers in real-time data, without making assumptions about the underlying noise distribution. This adaptability allows the framework to operate effectively in both Gaussian and non-Gaussian environments.

At the core of Conformal Prediction is the ability to construct prediction sets that reflect the uncertainty in the model's outputs. Specifically, CP quantifies non-conformity scores to identify outliers—samples that deviate significantly from the model's expected behavior. In the context of indoor localization, this capability is essential for detecting anomalies caused by noisy or anomalous signal patterns. By incorporating COD, our framework can dynamically adjust the size of the prediction set based on model confidence, offering a transparent view of the uncertainty in each prediction.

Our contributions are as follows:
\begin{itemize}
    \item \textbf{Adaptive Parameter Learning:} We propose a VB-UKF hybrid filter that simultaneously estimates state variables and noise covariances, enabling real-time adaptation to changing uncertainties.
    \item \textbf{Dual-Layer Outlier Detection:} We integrate Huber M-estimation to suppress mild outliers and COD with sliding-window calibration to provide statistical guarantees (coverage $\geq 1-\alpha$) for detecting severe anomalies.
    \item \textbf{Theoretical Rigor:} We provide proofs of semi-positive definiteness for Huber-weighted covariance matrices and finite-sample coverage bounds for non-exchangeable time series.
\end{itemize}

Experiments on geomagnetic fingerprint datasets demonstrate the effectiveness of our framework. Our method improves fingerprint matching accuracy from 81.25\% to 93.75\%, and reduces positioning errors from 0.62–6.87 m to 0.03–0.35 m. Additionally, the framework remains robust in the presence of non-Gaussian noise, achieving 95\% outlier detection precision with controlled false alarms.

The remainder of this paper is structured as follows: Section~\ref{RW} reviews related work on localization, filtering, and uncertainty quantification. Section~\ref{SD} presents the details of our proposed framework. Sections~\ref{EE}–~\ref{C} describe the experimental setup and results, followed by conclusions in Section~\ref{C}.

\section{Related Work}\label{RW}
\subsection{Indoor Localization Technologies}
Indoor positioning technologies primarily include Wi-Fi \cite{guo2022robust}, Bluetooth \cite{pau2021bluetooth}, \cite{sambu2022experimental}, geomagnetic \cite{sun2021indoor}, Inertial Measurement Units (IMU) \cite{sun2022indoor}, and Ultra-Wideband (UWB) \cite{qi2023calibration}. Wi-Fi and Bluetooth-based methods rely on fixed infrastructure, facing challenges in terms of coverage, maintenance costs, and security, which limit their widespread adoption \cite{ahmad2024recent}. UWB technology, with its advantages of low power consumption, long transmission range, and high data rate, has been widely applied in inertial measurement and attitude recognition. Geomagnetic and IMU-based methods, leveraging built-in sensors for passive positioning without external signals, exhibit high autonomy, making them promising solutions for achieving high-precision and cost-effective indoor positioning. However, while inertial navigation systems offer strong autonomy, their position and velocity errors accumulate rapidly over time. Similarly, geomagnetic sensors suffer from limited measurement accuracy and are susceptible to interference from currents and other magnetic fields, leading to measurement anomalies.

\subsection{Adaptive and Robust Filtering}
Single-sensor technologies face inherent limitations. Although multi-sensor data fusion can significantly improve positioning accuracy, it introduces modeling errors and typically assumes known noise distributions. To address these challenges, researchers have developed various adaptive and robust filtering techniques. Adaptive filters, such as fading memory filters \cite{kwon2022adaptive}, interactive multiple model filters \cite{fan2021interacting}, variational Bayesian (VB)-based filters \cite{lin2021variational}, and Gaussian sum filters \cite{chen2023gaussian}, excel in addressing modeling errors due to their strong recursive reasoning capabilities. However, these filters often perform poorly when measurement data contain outliers, as their estimates are easily contaminated, leading to error accumulation. Robust filters, such as those based on Huber’s M-estimation and federated filters \cite{wang2024improved}, are designed to handle large measurement errors and outliers.

For instance, Li et al. \cite{li2016variational} proposed an adaptive filter combining Gaussian-Newton iteration and VB approximation, capable of addressing both modeling errors and outliers in non-Gaussian scenarios. Liu et al. \cite{liu2021variational} introduced a robust variational Bayesian cubature Kalman filter (VBCKF) by incorporating the maximum correntropy criterion, which improved robustness to some extent. Davari et al. \cite{davari2021real} combined VB filters with hybrid artificial neural networks (ANN) to enable real-time outlier detection, but the high computational complexity and resource requirements limit their applicability in real-time scenarios. Yang et al. \cite{yang2024variational} proposed an adaptive robust filtering framework based on generalized maximum likelihood estimation (GM estimation) and VB methods, integrating the centered error entropy (CEE) criterion into the framework to develop the VBCEECKF algorithm. This approach enhanced robustness against outliers but required complex iterative computations, increasing deployment costs. Moreover, most existing filtering algorithms are optimized for specific noise distributions, making it challenging to maintain robustness across diverse noise settings.

\subsection{Uncertainty Quantification and Outlier Detection}
Uncertainty quantification (UQ) in localization systems has evolved significantly over the past decade. Early approaches primarily relied on probabilistic methods such as entropy estimation and Bayesian credible intervals, which derive uncertainty bounds from model outputs (e.g., softmax probabilities)~\cite{waymel2019impact}. While these methods provide intuitive uncertainty representations, their calibration often fails in practice due to overfitting or distributional mismatches, leading to overconfident or overly conservative predictions~\cite{shafer2008tutorial}.  

The emergence of conformal prediction (CP) marked a paradigm shift by offering distribution-free guarantees for uncertainty quantification. Unlike traditional methods, CP constructs prediction sets with rigorous coverage guarantees (\(1-\alpha\)) without requiring explicit likelihood models~\cite{shafer2008tutorial}. This model-agnostic framework has been widely adopted in classification tasks, where it quantifies uncertainty by evaluating the "non-conformity" of new observations against a calibration set. However, its direct application to time-series data---common in localization systems---faces inherent challenges due to violated exchangeability assumptions. Recent advances address this limitation through sliding window calibration~\cite{campos2024conformal} and sequential conformal inference~\cite{xu2023conformal}, which adaptively adjust thresholds to temporal dependencies while maintaining coverage guarantees. For instance, Xu~et~al.~\cite{xu2023conformal} demonstrated theoretically grounded prediction intervals for non-stationary time series, achieving robust performance in both simulated and real-world navigation tasks.  

In parallel, CP has been tailored to address outlier detection in dynamic environments. Strawn~et~al.~\cite{strawn2023conformal} integrated CP with reinforcement learning to design safety filters for collision avoidance, reducing collision rates by 80\% in robotic navigation. Similarly, Yang~et~al.~\cite{yang2023safe} leveraged CP to quantify perception uncertainty in LiDAR-based autonomous systems, achieving 93\% safety rates under stochastic sensor noise. These works highlight CP's versatility in bridging uncertainty quantification with real-time decision-making.  

Beyond robotics, CP has spurred innovations in sensor fusion and optimization. Garcia~et~al.~\cite{garcia2024multi} proposed multi-view conformal models, showing superior coverage performance over single-sensor baselines in heterogeneous environments. Meanwhile, Kim~et~al.~\cite{kim2025robust} addressed model misspecification in Bayesian optimization via localized conformal calibration, outperforming traditional Gaussian process-based methods in engineering design tasks.  

Despite these advancements, critical gaps persist in adapting CP to resource-constrained IoT localization systems. Existing studies often assume centralized computation or static environments, overlooking the challenges of real-time calibration in dynamic indoor settings with limited sensor bandwidth. Furthermore, while methods like sliding window CP~\cite{campos2024conformal} relax i.i.d. assumptions, their theoretical guarantees for non-linear state-space models---central to geomagnetic localization---remain underdeveloped.  

\section{System Design}\label{SD}
\subsection{Geomagnetic characteristics and data preprocessing}
Geomagnetic matching (MFM), as a passive positioning technology, utilizes the data of the Earth's magnetic field intensity and combines fingerprint matching for positioning. In this section, we collect geomagnetic data from different indoor locations to establish a geomagnetic fingerprint database, which is the main task in the offline phase. Then, we analyze its characteristics and determine the feasibility of using this technology for indoor positioning.

Liu et al. \cite{liu2018geomagnetism} pointed out that under the same path, the geomagnetic fingerprint data has a relatively high matching accuracy, that is, the collected geomagnetic signal time series has a relatively high similarity. Therefore, the path matching problem can be transformed into a time series matching problem. By collecting the geomagnetic intensity data of the same path from May to August 2024, we observed that the geomagnetic characteristics of the collected data remained consistent over a long period of time, even though the user heights and walking speeds were different, as shown in Figure~\ref{fig1}.

\begin{figure}[!h]
\centering
\includegraphics[width=0.5\textwidth]{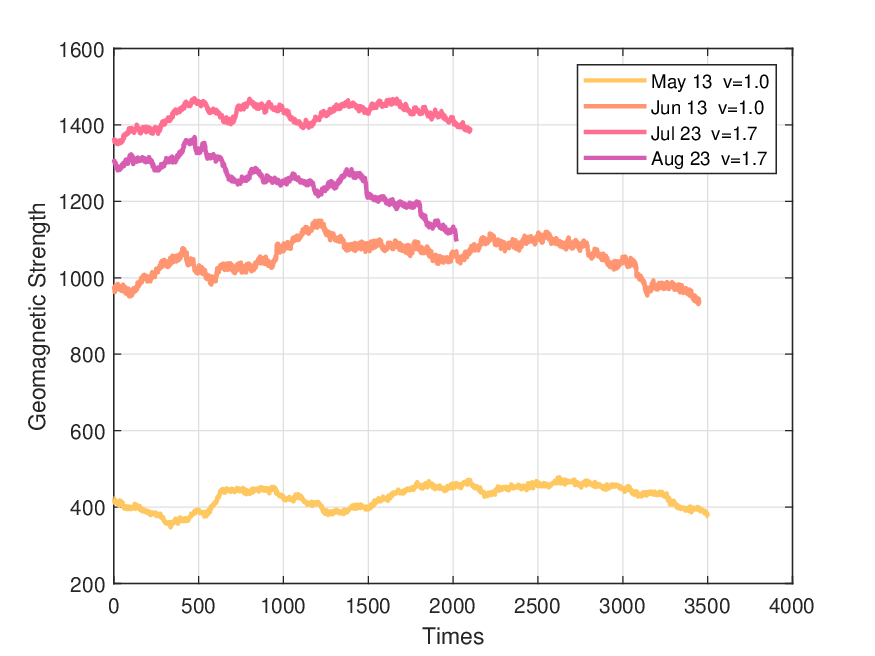}
\caption{Geomagnetic strength characteristics.}
\label{fig1}
\end{figure}

It is important to note that the geomagnetic data collection using geomagnetic sensors is performed in a carrier coordinate system (CCS), where the \(x\)-, \(y\)-, and \(z\)-axes correspond to the directions of the device's right side, front, and top, respectively. Therefore, analyzing the data collected in the carrier coordinate system is meaningless. It is necessary to convert the three-axis geomagnetic data into a global coordinate system (GCS). Let the magnetic field strength vector in the sensor's coordinate system be denoted as \(\mathbf{M}_{CCS}\). To transform it into the global coordinate system \(\mathbf{M}_{GCS}\), we apply three Euler angle-based rotation matrices, corresponding to rotations around the \(x\)-, \(y\)-, and \(z\)-axes:

\begin{equation}
R_{x}(\theta_{roll}) = \begin{bmatrix}
1 & 0 & 0 \\
0 & \cos(\theta_{roll}) & -\sin(\theta_{roll}) \\
0 & \sin(\theta_{roll}) & \cos(\theta_{roll})
\end{bmatrix}
\end{equation}

\begin{equation}
R_{y}(\theta_{pitch}) = \begin{bmatrix}
\cos(\theta_{pitch}) & 0 & \sin(\theta_{pitch}) \\
0 & 1 & 0 \\
-\sin(\theta_{pitch}) & 0 & \cos(\theta_{pitch})
\end{bmatrix}
\end{equation}

\begin{equation}
R_{z}(\theta_{yaw}) = \begin{bmatrix}
\cos(\theta_{yaw}) & -\sin(\theta_{yaw}) & 0 \\
\sin(\theta_{yaw}) & \cos(\theta_{yaw}) & 0 \\
0 & 0 & 1
\end{bmatrix}
\end{equation}

Besides, geomagnetic fingerprint data exhibits strong location dependence and is susceptible to fluctuations due to local environmental changes. In contrast, inertial sensor data tends to be stable over short periods but accumulates errors over time. By fusing these two data types, we can effectively leverage the global characteristics of geomagnetic fingerprints and the local characteristics of inertial data. This fusion significantly reduces the errors caused by inertial drift, achieving high-accuracy location estimation.

\subsection{Variational Bayesian Adaptive Filtering Framework}
In the online phase, the system continuously updates its state estimate through real-time processing of Received Signal Strength Indicator (RSSI) fingerprints, combining variational Bayesian learning with unscented transform to achieve robust nonlinear filtering. We first formalize the discrete-time state-space model:
\begin{equation}
\mathbf{x}_k = f(\mathbf{x}_{k-1}) + \mathbf{w}_{k-1}, \quad \mathbf{w}_{k-1} \sim \mathcal{N}(0,\mathbf{Q}_{k-1})
\end{equation}
\begin{equation}
\mathbf{z}_k = h(\mathbf{x}_k) + \mathbf{v}_k, \quad \mathbf{v}_k \sim \mathcal{N}(0,\mathbf{R}_k)
\end{equation}
where \(f(\cdot)\) and \(h(\cdot)\) are nonlinear state transition and observation functions. Assuming both the process noise \(\mathbf{w}_k\) and the measurement noise \(\mathbf{v}_k\) are zero-mean Gaussian variables simplifies the construction of the model. However, in practical scenarios, the measurement noise covariance matrix is often unknown or time-varying, which can significantly affect the accuracy of the estimates and reduce the system's robustness to outliers. As a result, Kalman filtering and its variants, which are effective for solving simplified models, may no longer provide accurate solutions in such cases.

To address these challenges, this paper focuses on adaptive filtering techniques that combine the Variational Bayesian (VB) approximation with the Adaptive Unscented Kalman Filter (AUKF), a method that has been demonstrated to enhance adaptability and robustness in \cite{li2016variational}. The specific steps of this framework are as follows:

1) \textbf{Sigma Point Generation}: For state vector \(\mathbf{x} \in \mathbb{R}^n\), generate \(2n+1\) sigma points:
\begin{equation}\label{1}
\begin{cases}
\xi_{0,k-1} = \hat{\mathbf{x}}_{k-1} \\
\xi_{i,k-1} = \hat{\mathbf{x}}_{k-1} + \left(\sqrt{(n+\kappa)\mathbf{P}_{k-1}}\right)_i, \quad i=1,\dots,n \\
\xi_{i,k-1} = \hat{\mathbf{x}}_{k-1} - \left(\sqrt{(n+\kappa)\mathbf{P}_{k-1}}\right)_{i-n}, \quad i=n+1,\dots,2n
\end{cases}
\end{equation}
where \(\kappa\) controls sigma point spread.

2) \textbf{State Prediction:}: Transform sigma points via \(f(\cdot)\):
\begin{equation}
\xi_{i,k} = f(\xi_{i,k-1}), \quad i=0,\dots,2n
\end{equation}

The predicted mean and covariance:
\begin{equation}\label{2}
\hat{\mathbf{x}}_{k|k-1} = \sum_{i=0}^{2n} \omega_i^{(m)} \xi_{i,k}, 
\end{equation}

\begin{equation}\label{3}
\quad \mathbf{P}_{k|k-1} = \sum_{i=0}^{2n} \omega_i^{(c)} (\xi_{i,k}-\hat{\mathbf{x}}_{k|k-1})(\cdot)^T + \mathbf{Q}_{k-1}
\end{equation}
with weights \(\omega_i^{(m)}\), \(\omega_i^{(c)}\) determined by \(\kappa\).

3) \textbf{Variational Parameter Prediction}: For measurement noise covariance \(\mathbf{R}_k \sim \mathcal{IW}(v_k, V_k)\) (inverse-Wishart):
\begin{equation}\label{4}
v_{k|k-1} = \rho(v_{k-1} - n - 1) + n + 1, 
\end{equation}

\begin{equation}\label{5}
\quad V_{k|k-1} = \sqrt{\rho} V_{k-1} \sqrt{\rho}^T
\end{equation}
where \(\rho \in (0,1]\) governs temporal correlation.

The variational Bayesian framework aims to approximate the true posterior$p(\mathbf{x}_k, \mathbf{R}_k | \mathbf{z}_{1:k})$ by maximizing the Kullback-Leibler (KL) divergence between the variational distribution$q(\mathbf{x}_k,\mathbf{R}_k)$ and the true posterior. This is equivalent to maximizing the Evidence Lower Bound (ELBO):
\begin{equation}
\begin{split}
    \text{ELBO} = \mathbb{E}_q \left[ \log p(\mathbf{z}_k | \mathbf{x}_k, \mathbf{R}_k) \right] + \mathbb{E}_q \left[ \log p(\mathbf{x}_k | \mathbf{x}_{k-1}) \right] - \\
    \text{KL} \left( q(\mathbf{R}_k) \| p(\mathbf{R}_k) \right) - \text{KL} \left( q(\mathbf{x}_k) \| p(\mathbf{x}_k | \mathbf{x}_{k-1}) \right),
\end{split}
\end{equation}
where $p(\mathbf{R}_k)$ is the inverse-Wishart prior and $q(\mathbf{R}_k)$ is the variational posterior.

The inverse-Wishart distribution is chosen as the conjugate prior for the measurement noise covariance $\mathbf{R}_k$, ensuring closed-form updates during variational inference. This avoids complex numerical integration and maintains computational efficiency.

4) \textbf{Measurement Update}:
\begin{equation}
\zeta_{i,k}^{(j)} = h(\xi_{i,k}^{(j)})
\end{equation}

\begin{equation}
\quad \hat{\mathbf{z}}_{k}^{(j)} = \sum_{i=0}^{2n} \omega_i^{(m)} \zeta_{i,k}^{(j)}
\end{equation}

5) \textbf{Covariance Matrices}:
\begin{align}
\mathbf{P}_{zz}^{(j)} &= \sum_{i=0}^{2n} \omega_i^{(c)} \left( \zeta_{i,k}^{(j)} - \hat{\mathbf{z}}_{k}^{(j)} \right) \left( \zeta_{i,k}^{(j)} - \hat{\mathbf{z}}_{k}^{(j)} \right)^T + \frac{V_k^{(j)}}{v_k - n - 1}, \\
\mathbf{P}_{xz}^{(j)} &= \sum_{i=0}^{2n} \omega_i^{(c)} \left( \xi_{i,k}^{(j)} - \hat{\mathbf{x}}_{k}^{(j)} \right) \left( \zeta_{i,k}^{(j)} - \hat{\mathbf{z}}_{k}^{(j)} \right)^T.
\end{align}
The term $\frac{V_k^{(j)}}{v_k - n - 1}$ represents the expectation of the measurement noise covariance $\mathbf{R}_k$ under the inverse-Wishart variational posterior $q(\mathbf{R}_k)$, ensuring adaptive noise estimation.

6) \textbf{State \& Covariance Update}:
\begin{equation}
\mathbf{K}_k^{(j)} = \mathbf{P}_{xz}^{(j)} (\mathbf{P}_{zz}^{(j)})^{-1}
\end{equation}

\begin{equation}
\quad \hat{\mathbf{x}}_k^{(j+1)} = \hat{\mathbf{x}}_k^{(j)} + \mathbf{K}_k^{(j)} (\mathbf{z}_k - \hat{\mathbf{z}}_k^{(j)})
\end{equation}

7) \textbf{Noise Parameter Update}:
\begin{equation}
V_k^{(j+1)} = V_k^{(j)} + \sum_{i=0}^{2n} \omega_i^{(c)} (\mathbf{z}_k - \zeta_{i,k}^{(j)})(\cdot)^T
\end{equation}

Iterate steps 4)-7) for \(N\) cycles to approximate the joint posterior \(p(\mathbf{x}_k, \mathbf{R}_k|\mathbf{z}_{1:k})\).

\begin{algorithm}[t]
\caption{Variational Bayesian AUKF}
\begin{algorithmic}[1]
\REQUIRE $\hat{\mathbf{x}}_{k-1}$, $\mathbf{P}_{k-1}$, $\mathbf{z}_k$, $f(\cdot)$, $h(\cdot)$, $\mathbf{Q}_{k-1}$, $\mathbf{R}_{k-1}$, $v_{k-1}$, $V_{k-1}$
\STATE Generate sigma points $\xi_{i,k-1}$ via \eqref{1}
\STATE Predict $\hat{\mathbf{x}}_{k|k-1}$ and $\mathbf{P}_{k|k-1}$ using \eqref{2}-\eqref{3}
\STATE Predict $v_{k|k-1}$ and $V_{k|k-1}$ via \eqref{4}-\eqref{5}
\FOR{$j=1$ to $N$}
    \STATE Compute measurement sigma points $\zeta_{i,k}^{(j)}$
    \STATE Update covariance matrices $\mathbf{P}_{zz}^{(j)}$, $\mathbf{P}_{xz}^{(j)}$
    \STATE Calculate Kalman gain $\mathbf{K}_k^{(j)}$
    \STATE Update state estimate $\hat{\mathbf{x}}_k^{(j+1)}$
    \STATE Update noise parameter $V_k^{(j+1)}$
\ENDFOR
\ENSURE $\hat{\mathbf{x}}_k^{(N)}$, $\mathbf{P}_k^{(N)}$, $V_k^{(N)}$
\end{algorithmic}
\end{algorithm}

The framework provides three fundamental advantages:
\begin{itemize}
\item \textbf{Adaptive Noise Estimation}: Joint state and noise covariance estimation through variational inference
\item \textbf{Nonlinear Handling}: Unscented transform maintains second-order accuracy without linearization
\item \textbf{Computational Efficiency}: Fixed iteration count ensures real-time operation
\end{itemize}

\subsection{Hierarchical Outlier Handling Mechanism}

\subsubsection{Huber Robust Update}
In M-estimation, the predictor update is often framed as a weighted least squares problem. By constructing a nonlinear regression measurement model and minimizing the generalized Huber loss function, the corrected measurement values can be derived. 

To mitigate the impact of mild outliers on state estimation, a Huber loss function is adopted to construct the weighting matrix. The piecewise weighting function is defined as:
\begin{equation}
W(r_k) = 
\begin{cases} 
1, & |r_k| \leq \delta \\
\delta / |r_k|, & |r_k| > \delta 
\end{cases}
\label{eq:huber_weight}
\end{equation}
where \( r_k = z_k - H\hat{x}_{k|k-1} \) is the measurement residual, and \( \delta \) denotes the robustness threshold. A recommended value is \( \delta = 1.345\sigma \), where \( \sigma \) represents the standard deviation of the observation noise.

The Huber weighting is incorporated into the standard Kalman update, resulting in the modified Kalman gain:

\begin{equation}
K_k = P_{k|k-1}H^T \left( H P_{k|k-1}H^T + R/W(r_k) \right)^{-1}
\label{eq:modified_gain}
\end{equation}

The term \( R/W(r_k) \) adaptively downweights anomalous observations: when \( |r_k| > \delta \), \( W(r_k) < 1 \), thereby increasing the effective measurement noise covariance.

\begin{theorem}[Huber-Weighted Posterior Covariance]
\label{thm:huber_cov}
Given the Huber-modified posterior covariance matrix:
\begin{equation}
P_{k|k} = (I - W_k K_k H)P_{k|k-1}(I - W_k K_k H)^T + W_k K_k R K_k^T
\end{equation}
where:
\begin{itemize}
\item $W_k \geq 0$ is the Huber weighting factor
\item $P_{k|k-1} \succ 0$ is the prior covariance
\item $R \succ 0$ is the measurement noise covariance
\end{itemize}
Then $P_{k|k}$ remains positive semi-definite, i.e., $P_{k|k} \succeq 0$.
\end{theorem}

Details on the proof process are provided in Appendix A.

\begin{algorithm}[htbp]
\caption{Huber Robust Update Procedure}
\begin{algorithmic}[1]
\REQUIRE Predicted state \(\hat{x}_{k|k-1}\), predicted covariance \(P_{k|k-1}\), measurement \(z_k\)
\STATE Compute residual \( r_k \leftarrow z_k - H\hat{x}_{k|k-1} \)
\STATE Calculate weight \(W(r_t)\) via (\ref{eq:huber_weight})
\STATE Compute modified gain \(K_t\) via (\ref{eq:modified_gain})
\STATE Update state \(\hat{x}_{k|k} \leftarrow \hat{x}_{k|k-1} + W(r_t)K_t r_t \)
\STATE Update covariance \(P_{k|k} \leftarrow (I - W(r_t)K_t H)P_{k|k-1} + W(r_t)K_t R K_t^T \)
\ENSURE Posterior state \(\hat{x}_{k|k}\), covariance \(P_{k|k}\)
\end{algorithmic}
\end{algorithm}

\subsubsection{Conformal Outlier Detection}
While the Huber function provides robustness against moderate outliers, its dependence on a static threshold limits adaptability to dynamic noise environments. To address this, we introduce a conformal outlier detection mechanism as a secondary defense layer, offering statistical guarantees for error control.

This section begins with a formulation of the indoor positioning problem, followed by an overview of the fundamental CP framework. Finally, we incorporate conformal outlier detection to enhance the robustness of the filter.

\paragraph{Problem Formulation}
The fingerprint-based indoor localization is first modeled as a classification task. Given an offline dataset of $n$ i.i.d. RSSI-location pairs $\{(X_i, Y_i)\}_{i=1}^n$, let $\hat{f}: \mathcal{X} \to \Delta^K$ be a probabilistic predictor that outputs a distribution over $K$ candidate locations, where $\Delta^K$ denotes the probability simplex. For each test instance $X_{\text{test}}$, we quantify its conformity to the training distribution through a non-conformity score.

\paragraph{Conformal Prediction Framework}
Define $s: \mathcal{X} \times \mathcal{Y} \to \mathbb{R}$ as the non-conformity score:
\begin{equation}
    s(X_i, Y_i) = 1 - \hat{f}(X_i)_{Y_i}
\end{equation}
where $\hat{f}(X_i)_{Y_i}$ denotes the predicted probability for the true class $Y_i$.

Compute the $\lceil (n+1)(1-\alpha) \rceil$-th quantile as conformal threshold:
\begin{equation}
    \hat{q} = \inf\left\{q : \frac{|\{i : s_i \leq q\}|}{n} \geq \frac{\lceil (n+1)(1-\alpha) \rceil}{n}\right\}
\end{equation}

Then we can build the prediction set:
\begin{equation}
    C(X_{\text{test}}) = \left\{y \in \{1,\ldots,K\} : s(X_{\text{test}}, y) \leq \hat{q}\right\}
\end{equation}

This construction guarantees marginal coverage:
\begin{equation}
    \mathbb{P}\left(Y_{\text{test}} \in C(X_{\text{test}})\right) \geq 1 - \alpha
\end{equation}

\paragraph{Extension to Filtering Framework} 
For sequential localization tasks, we adapt the conformal framework to handle time-series data:

First we maintain a sliding window of $w$ recent samples $\mathcal{C}_k = \{s_{k-w}, \ldots, s_{k-1}\}$ to build a dynamic calibration set.

Then define the non-conformity score as Filter-Aware Scoring using normalized innovation:
\begin{equation}
    s_k = \|\mathbf{z}_k - H\hat{\mathbf{x}}_{k|k-1}\|_{\mathbf{S}_k^{-1}}
\end{equation}
\begin{equation}
    \quad \mathbf{S}_k = H\mathbf{P}_{k|k-1}H^\top + R_k
\end{equation}

Finally, the adaptive thresholding:
\begin{equation}
    \hat{q}_k = \text{Quantile}\left(\mathcal{C}_k; \frac{\lceil (w+1)(1-\alpha) \rceil}{w}\right)
\end{equation}

\paragraph{Outlier Handling Mechanism}
When $s_k > \hat{q}_k$, indicating statistical inconsistency:
\begin{equation}
    R_k \leftarrow \gamma R_k \quad (\gamma > 1)
\end{equation}
This covariance inflation strategy temporarily reduces outlier influence while maintaining filter stability.

\begin{theorem}[Time-Varying Coverage]
Under exchangeability within the sliding window, the adaptive threshold $\hat{q}_k$ satisfies:
\begin{equation}
    \mathbb{P}(s_k \leq \hat{q}_k) \geq 1 - \alpha - \mathcal{O}\left(\frac{1}{\sqrt{w}}\right)
\end{equation}
\end{theorem}
Proof appears in Appendix B.

\paragraph{Effectiveness}
\begin{figure}[H]
\centering
\includegraphics[width=0.5\textwidth]{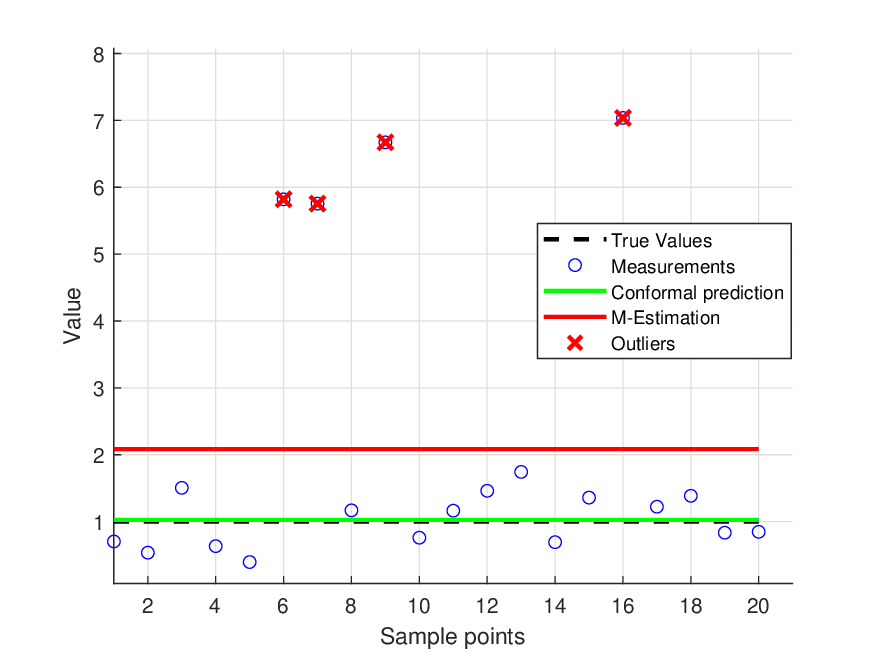}
\caption{Conformal outlier detection results with $\alpha=0.05$.}
\label{fig2}
\end{figure}
To evaluate the effectiveness of this approach, consider a scenario where we have \(n = 20\) test samples and need to determine whether a new test point is an outlier. First, we normalize the calculation results of non-conformity score(NS) and sort the NSs in ascending order. The closer the score is to 1, the higher the probability that the test point is an outlier. The objective is to identify as many outliers as possible while ensuring that no more than \(\alpha = 0.05\) of the non-outliers are incorrectly classified as outliers. Figure~\ref{fig2} illustrates the results of this procedure using an outlier detection model. 

\section{Experiments and Evaluation}\label{EE}
\subsection{Numerical Simulation}
To effectively evaluate the accuracy and robustness of the designed filter, this study selects a classical nonlinear numerical case, namely the Univariate Nonstationary Growth Model (UNGM), to demonstrate the filter's performance. The discrete-time state-space equation of the UNGM can be expressed as:
\begin{equation}
x_{k} = 0.5x_{k - 1} + \frac{25x_{k - 1}}{1 + x_{k - 1}^{2}} + 8\cos(1.2(k - 1)) + v_{k - 1}
\end{equation}
\begin{equation}
y_{k} = \frac{x_{k}^{2}}{20} + r_{k}, \quad k = 1, 2, \dots, M
\end{equation}
Here, \(v_{k - 1}\) represents a Gaussian process noise with zero mean and covariance \(Q_{k - 1}\). \(r_{k}\) denotes the measurement noise. The index \(k\) represents the simulation step, and \(M\) is the total number of simulation steps. One term in the equation is independent of \(x_{k}\), but varies with time \(k\), which can be interpreted as time-varying noise.
\begin{figure}[h]
\centering
\includegraphics[width=0.45\textwidth]{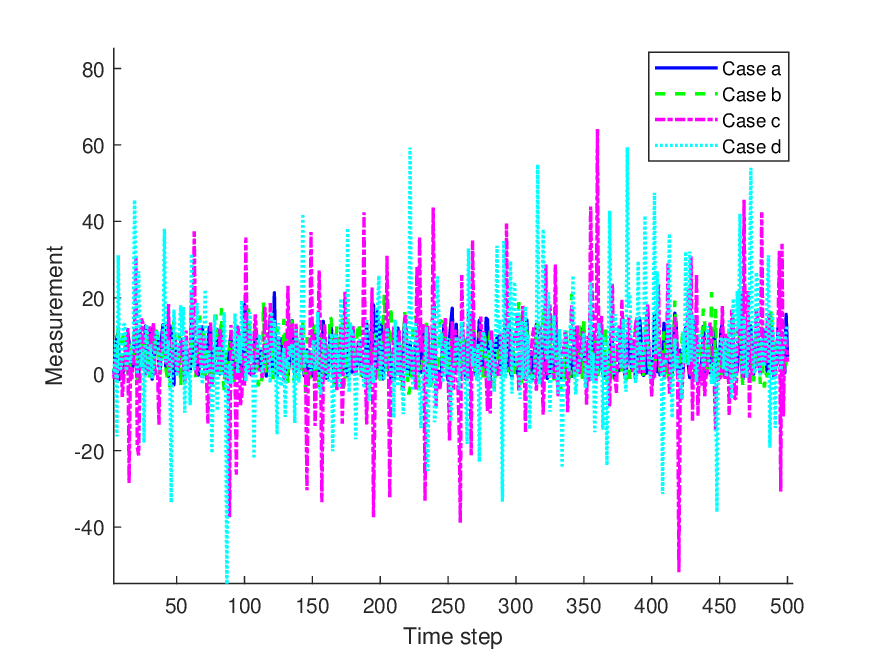}
\caption{Comparison of measurements with different noise.}
\label{fig3}
\end{figure}
To demonstrate the performance of the filter, simulations are conducted under the following four different types of noise scenarios:
\begin{enumerate}
    \item Case a: Measurement noise is known and follows a Gaussian distribution, \(r_{k} \sim \mathcal{N}(0, R_{k})\), where \(k = 1\).
    \item Case b: The covariance of the measurement noise is time-varying and follows a Gaussian distribution, \(r_{k}^{c} \sim \mathcal{N}(0, R_{k}^{c})\), with \(R_{k}^{c}\) given by:
        \begin{equation}
        R_{k}^{c} = \begin{cases}
        3 + 2(2 + \arctan(0.3(k - M/4))), & k < 1.5 \times M/4 \\
        3 + 2(2 + \arctan(-0.3(k - M/2))), & k > 1.5 \times M/4
        \end{cases} 
        \end{equation}
    \item Case c: Measurement noise follows a mixture of Gaussian istributions:
        \begin{equation}
        r_{k}^{c} \sim (1 - a)\mathcal{N}(r_{1,k} | 0, R_{1}^{2}) + a \mathcal{N}(r_{2,k} | 0, R_{2}^{2})
        \end{equation}
    where $a$ is a disturbance factor representing the extent of contamination in the error model, and $R_{1}$ and $R_{2}$ are the standard deviations of the Gaussian distributions. The first Gaussian distribution is a standard Gaussian, while the second distribution is considered a perturbation distribution. The larger the value of \(a\), the more significant the disturbance.
    \item Case d: This case combines both b and c, with measurement noise following a Gaussian mixture distribution \(r_{k}^{c}\), where the variance is time-varying \(R_{1}^{c} = R_{k}^{c}\):
        \begin{equation}
        r_{k}^{c} \sim (1 - a)\mathcal{N}(r_{1,k}^{c} | 0, R_{k}^{2}) + a \mathcal{N}(r_{2,k} | 0, R_{2}^{2})
        \end{equation}
\end{enumerate}

The measurement values are shown in Figure~\ref{fig3}.
And the simulation results of VB-based filter and CP-based filer for four different measurement noises are shown in Figure~\ref{fig4}. It can be observed that CP method can effectively improve the filter effect when the noise is case a and case b, while for c and d, although the improvement effect is not significant, it further inhibits the influence of outliers on the basis of the original robust algorithm.

We selected two classical filters, namely the Particle Filter (PF) and Unscented Kalman Filter (UKF), a robust Huber-based UKF (HUKF), and two adaptive robust filters for the simulation. The results are presented in Figure~\ref{fig5}. 
\begin{figure}[h]
\centering
\subfloat[]{%
    \includegraphics[width=0.5\linewidth]{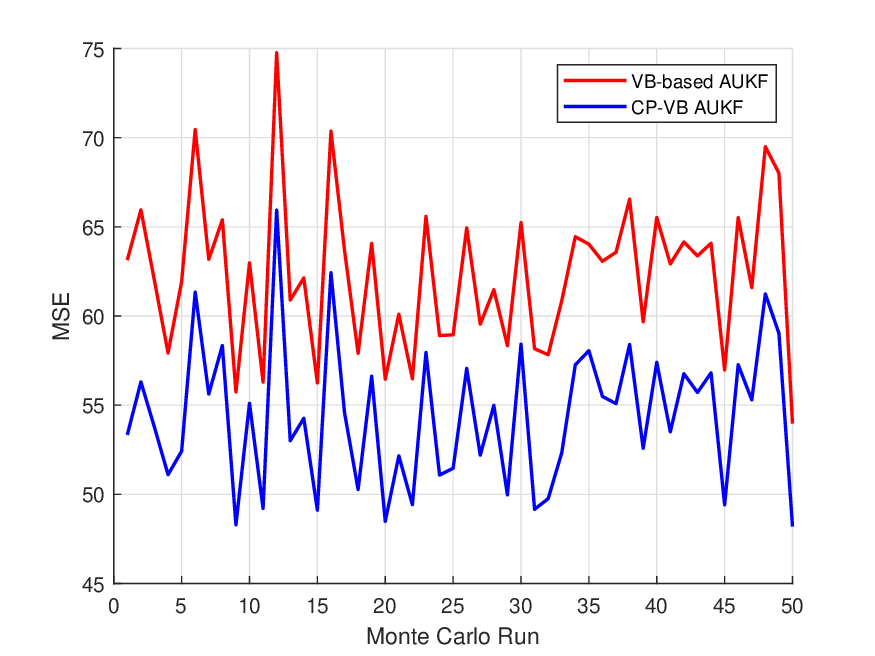}%
    \label{fig4a}%
} 
\subfloat[]{%
    \includegraphics[width=0.5\linewidth]{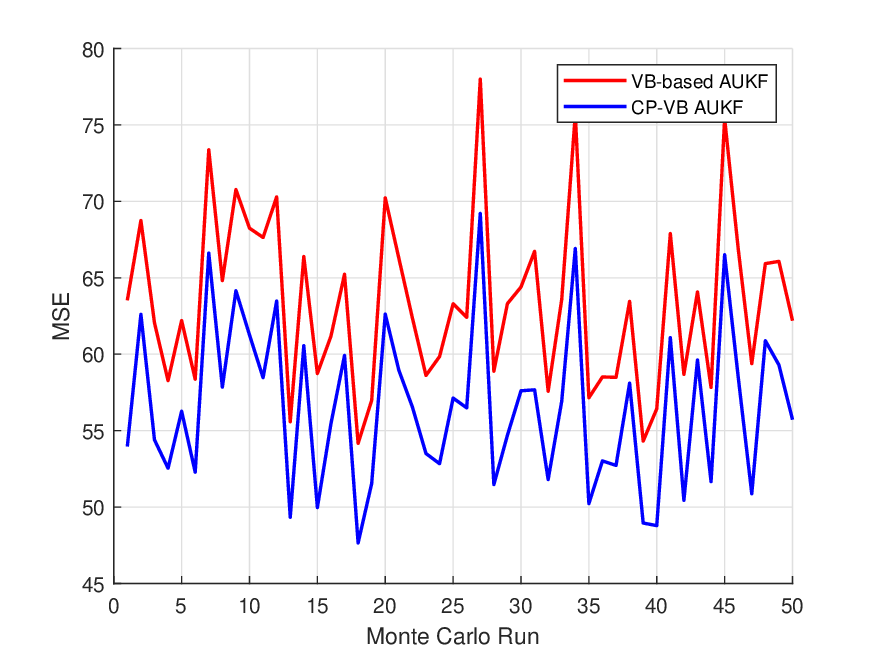}%
    \label{fig4b}%
} \hfill
\subfloat[]{%
    \includegraphics[width=0.5\linewidth]{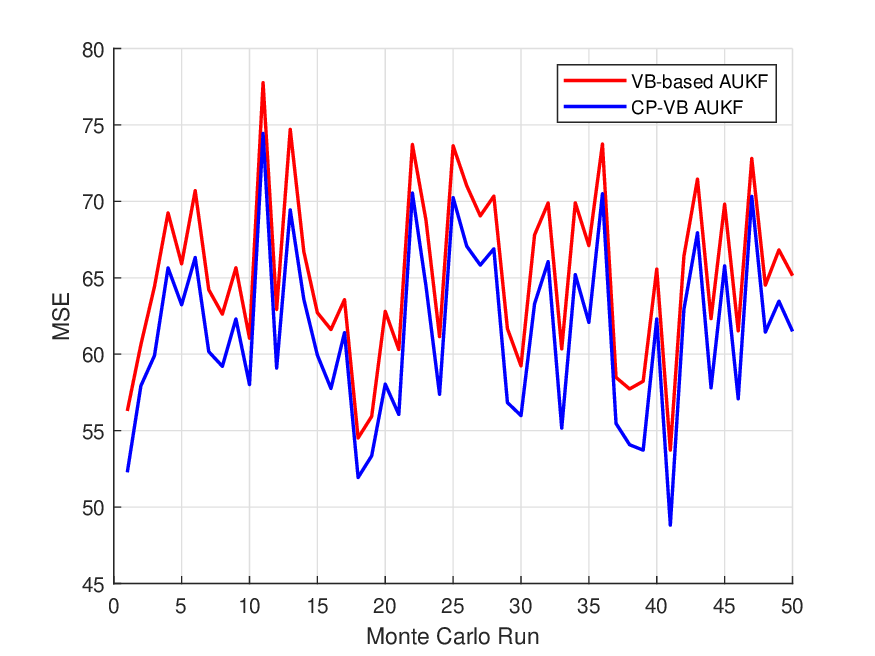}%
    \label{fig4c}%
}
\subfloat[]{%
    \includegraphics[width=0.5\linewidth]{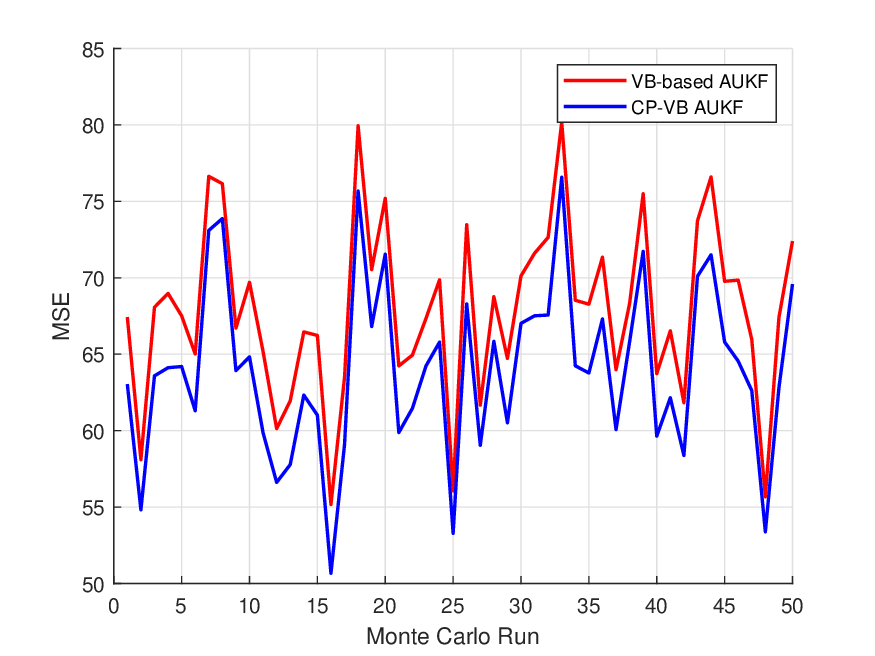}%
    \label{fig4d}%
}
\caption{The simulation results of VB-based AUKF filters for the four different types of measurement noise. (a) Case a. (b) Case b. (c) Case c. (d) Case d.}
\label{fig4}
\end{figure}
It is evident that as the impact of outliers increases, the advantage of adaptive robust filters becomes more pronounced. After applying the conformal outlier detection method proposed in this paper, all filters show improved performance. Interestingly, even the classical filters, which typically perform worse than adaptive robust filters, are able to effectively mitigate the influence of measurement noise with contaminated distributions.

To quantitatively validate the robustness improvements, Table~\ref{tb1} compares the Mean Squared Error (MSE) of five filters under four noise scenarios, both with and without conformal outlier detection (We mark it simply as CP in the figures, and in the following text as well) enhancement. Three key observations emerge:

\begin{itemize}
    \item \textbf{Universal Performance Gain}: CP integration reduces MSE across all filters and noise types. The most significant improvement occurs in Case~c (non-Gaussian mixture), where UKF's MSE decreases by 53.2\% (200.70$\rightarrow$93.89), demonstrating CP's effectiveness against heavy-tailed distributions.
    
    \item \textbf{Classical vs. Adaptive Filters}: While adaptive filters (VB-AUKF/VB-HAUKF) achieve lower baseline MSE, classical filters show greater relative improvement with CP. For instance, PF achieves 21.3\% MSE reduction in Case~d versus 1.3\% for VB-HAUKF, suggesting CP compensates for classical methods' weaker inherent robustness.
    
    \item \textbf{Noise-Type Dependency}: CP provides maximum benefit in time-varying non-Gaussian scenarios (Cases~c/d), reducing average MSE by 46.7\% compared to 21.9\% in stationary Gaussian cases (Cases~a/b). This aligns with the visual trends of Figure~\ref{fig4} and Figure~\ref{fig5}.
\end{itemize}

Notably, the proposed VB-HAUKF+CP combination achieves the lowest absolute MSE in all scenarios (49.16$\sim$98.88), outperforming second-best VB-AUKF+CP by 6.1\%$\sim$9.4\%. This demonstrates the synergistic effect of variational Bayesian adaptation and CP-based outlier rejection.

\begin{figure}[h]
\centering
\subfloat[]{%
    \includegraphics[width=0.5\linewidth]{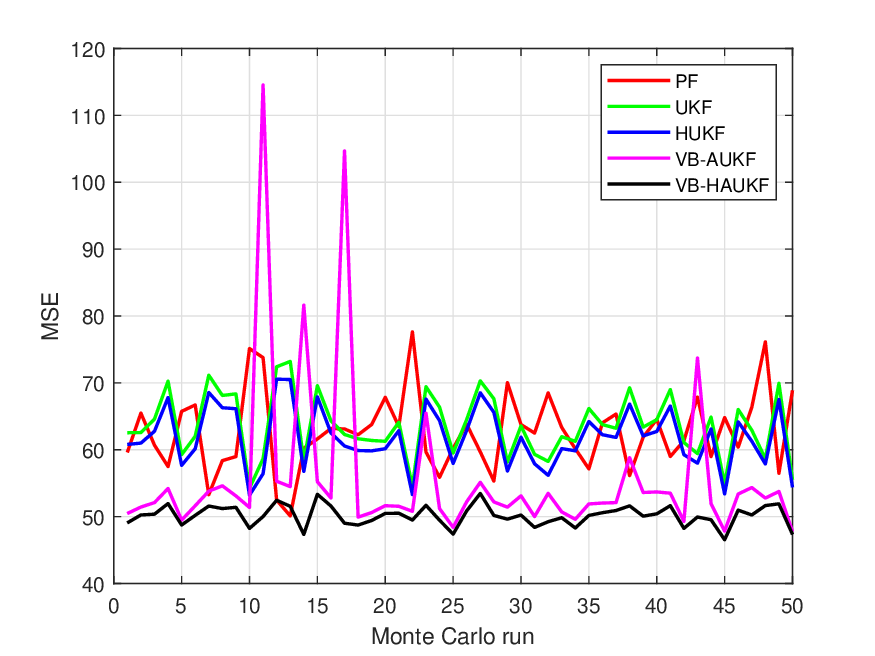}%
    \label{fig5a}%
} 
\subfloat[]{%
    \includegraphics[width=0.5\linewidth]{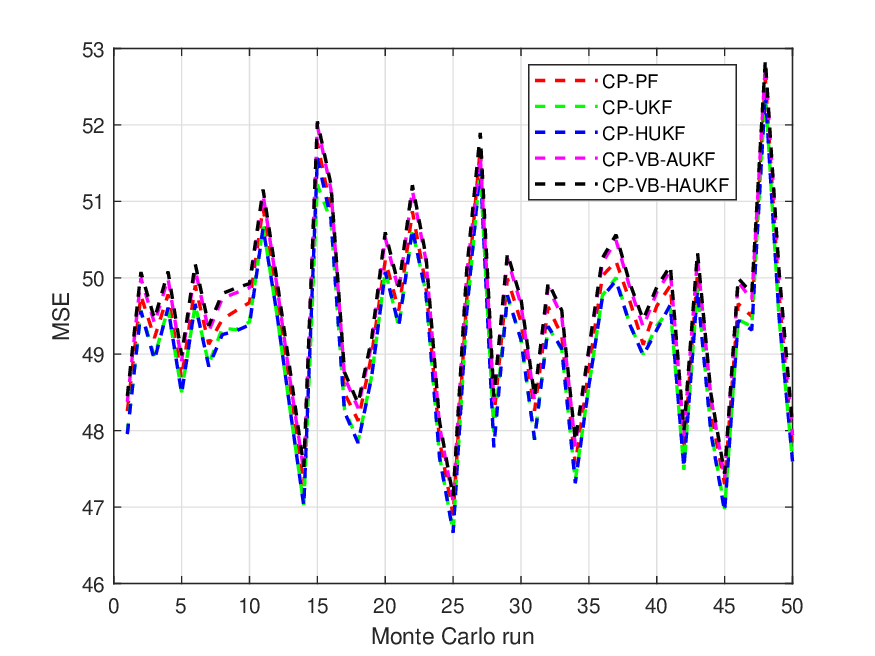}%
    \label{fig5b}%
} \hfill
\subfloat[]{%
    \includegraphics[width=0.5\linewidth]{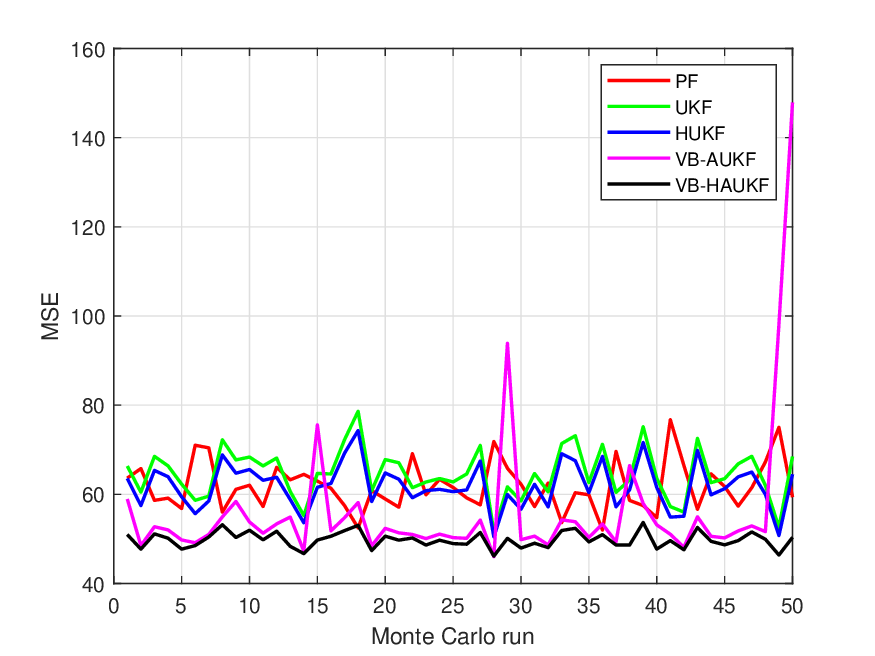}%
    \label{fig5c}%
}
\subfloat[]{%
    \includegraphics[width=0.5\linewidth]{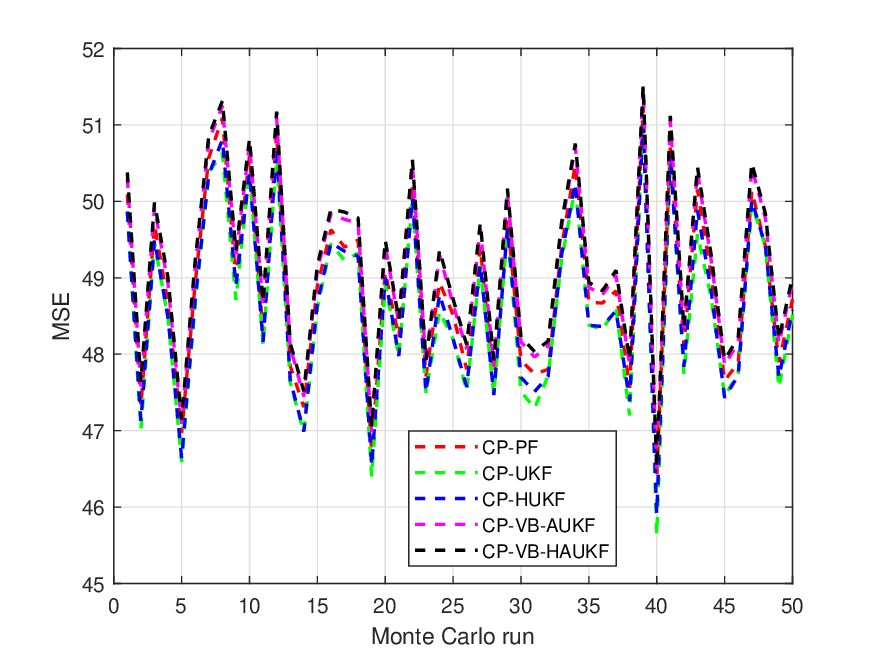}%
    \label{fig5d}%
} \hfill
\subfloat[]{%
    \includegraphics[width=0.5\linewidth]{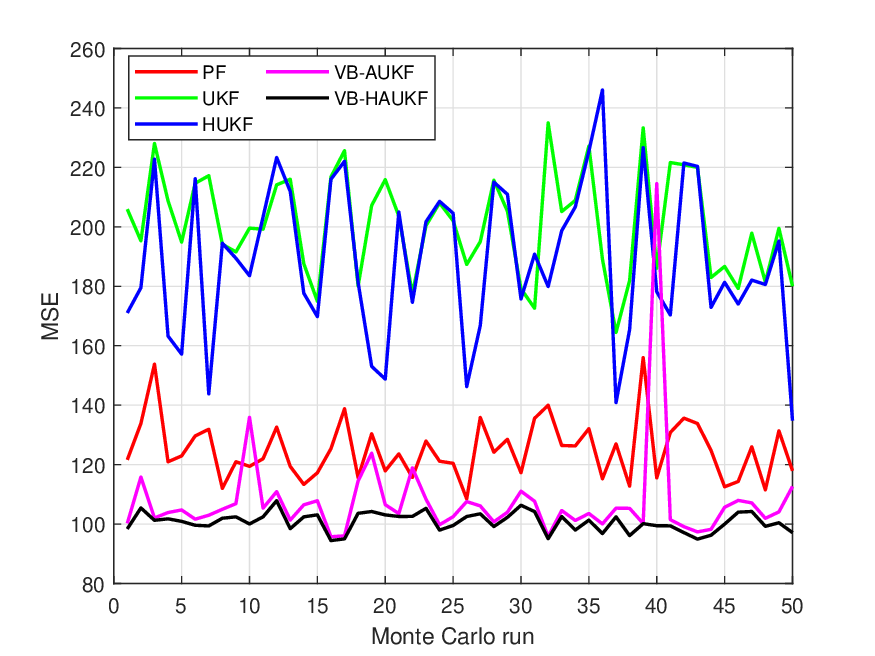}%
    \label{fig5e}%
} 
\subfloat[]{%
    \includegraphics[width=0.5\linewidth]{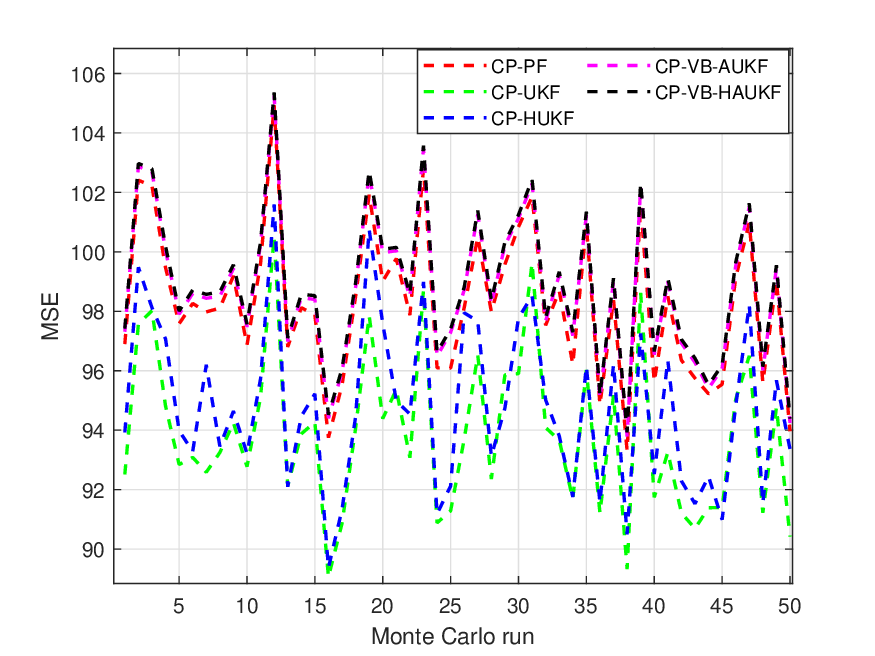}%
    \label{fig5f}%
} \hfill
\subfloat[]{%
    \includegraphics[width=0.5\linewidth]{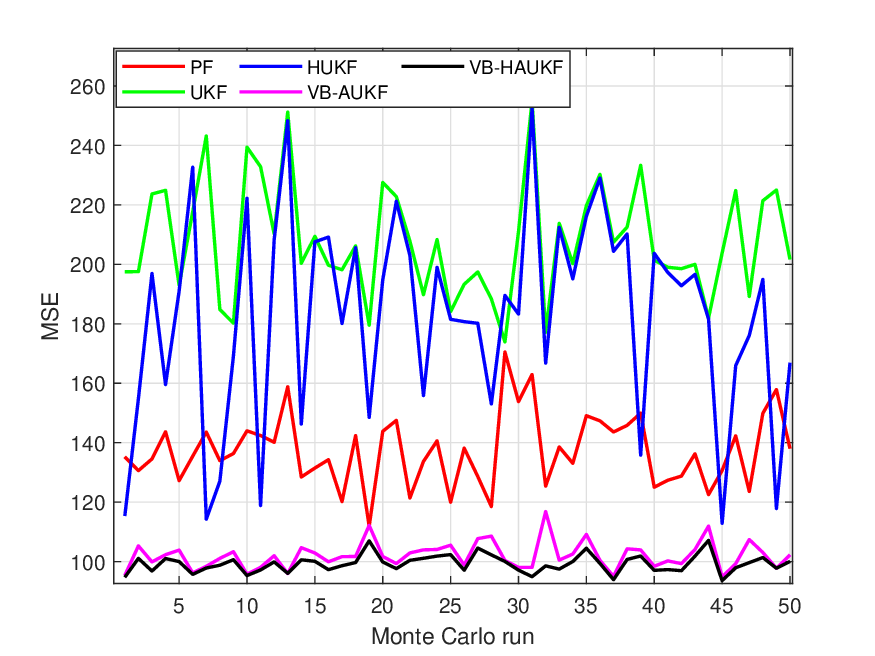}%
    \label{fig5g}%
}
\subfloat[]{%
    \includegraphics[width=0.5\linewidth]{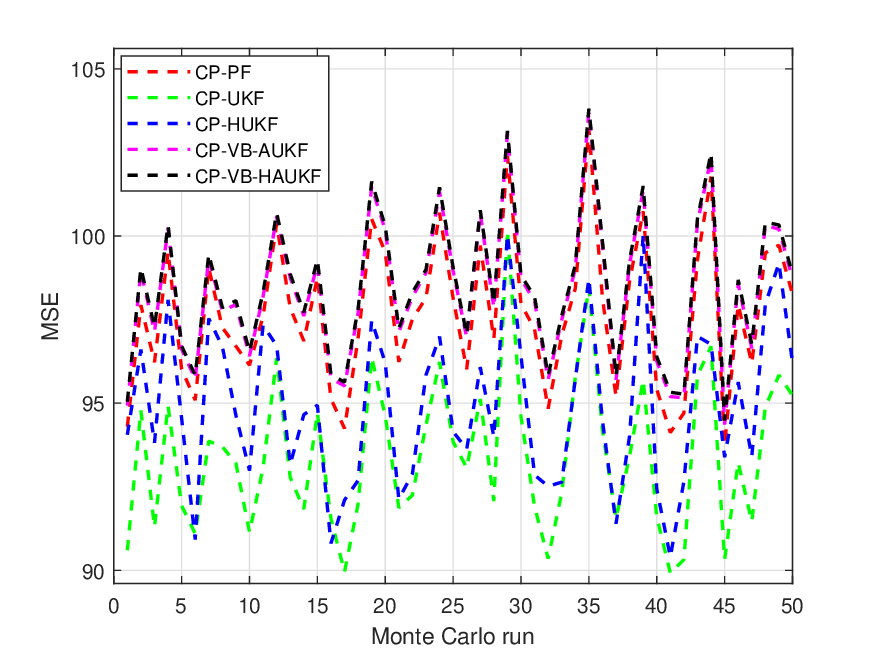}%
    \label{fig5h}%
}
\caption{The simulation results of different filters for the four different types of measurement noise. (a) Case a. (b) Case b. (c) Case c. (d) Case d.}
\label{fig5}
\end{figure}

\begin{table}[!t]
    \centering
    \caption{The mean MSE of the five filters}
    \resizebox{\linewidth}{!}{\begin{tabular}{ccccccccc} 
    \toprule
    \multirow{3}{*}{filters} & \multicolumn{8}{c}{Mean-MSE}   \\
    \cline{2-9}  
    & \multicolumn{2}{c}{Case a} & \multicolumn{2}{c}{Case b} & \multicolumn{2}{c}{Case c} & \multicolumn{2}{c}{Case d} \\
    \cline{2-9}  
    & original & CP-based & original & CP-based & original & CP-based & original & CP-based  \\
    \midrule
    PF & 62.66 & 49.35 & 61.81 & 48.86 & 125.00 & 98.27 & 137.40 & 97.68 \\
    UKF & 63.55 & 49.11 & 64.62 & 48.60 & 200.70 & 93.89 & 207.8 & 93.45 \\
    HUKF & 61.81 & 49.11 & 61.91 & 48.66 & 188.60 & 94.79 & 182.5 & 94.91 \\
    VB-AUKF & 55.77 & 49.55 & 56.41 & 49.10 & 107.70 & 98.77 & 102 & 98.39 \\
    VB-HAUKF & 50.15 & 49.62 & 49.79 & 49.16 & 100.70 & 98.88 & 99.28 & 98.51 \\
    \bottomrule
\end{tabular}}
\label{tb1}
\end{table}

To visually illustrate the effectiveness of the filters, Figure~\ref{fig6} shows the cumulative distribution function (CDF) curves of the four filters under different conditions, which reveal three critical insights about filter performance:

\begin{itemize}
    \item \textbf{Tail Behavior Characterization}: 
    In heavy-tailed scenarios (Cases~c/d), the proposed CP-based combination exhibits the steepest CDF ascent within the low-error region (e.g., $<100$ in Fig.~\ref{fig6c}), with 90\% of errors below 95 versus 107 for vanilla UKF. This indicates effective suppression of outlier-induced large errors. The CP-enhanced curves (dashed lines) show compressed right tails compared to original filters (solid lines), particularly visible in 95th percentile range.

    \item \textbf{Consistency Across Noise Types}:
    While adaptive filters (VB-AUKF/VB-HAUKF) maintain tight error distributions in Gaussian cases (Cases~a/b), CP brings classical filters (UKF) closer to their performance. For instance, in Case~a, UKF+CP reaches 95\% CDF at 51 error—even higher slightly than VB-HAUKF+CP, despite UKF's baseline being 71 worse.

    \item \textbf{Algorithmic Robustness Trade-off}:
    The HUKF's CDF curves demonstrate median error robustness but falter in extreme cases—its 95th percentile error in Case~c reaches more than 220 versus 100 for VB-HAUKF+CP. This quantifies the trade-off between heuristic robustness (HUKF) and learning-based adaptation (VB-HAUKF+CP).
\end{itemize}

These observations align with the MSE metrics in Table~\ref{tb1} while providing new dimensionality: the CDFs quantify not just average performance but \textit{error distribution shaping}—critical for safety-sensitive applications where bounding worst-case errors is paramount.

\begin{figure}[h]
\centering
\subfloat[]{%
    \includegraphics[width=0.5\linewidth]{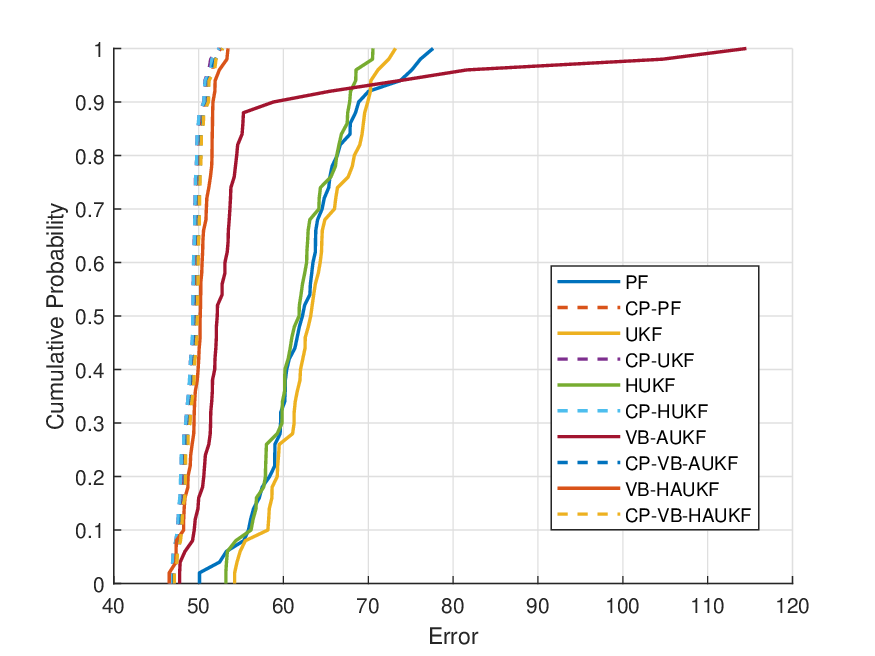}%
    \label{fig6a}%
} 
\subfloat[]{%
    \includegraphics[width=0.5\linewidth]{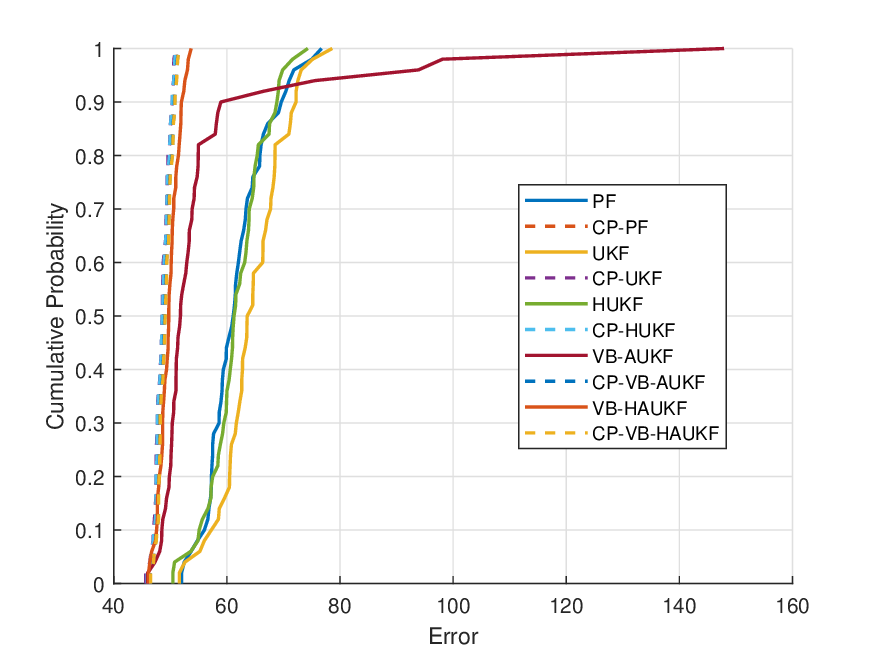}%
    \label{fig6b}%
} \hfill
\subfloat[]{%
    \includegraphics[width=0.5\linewidth]{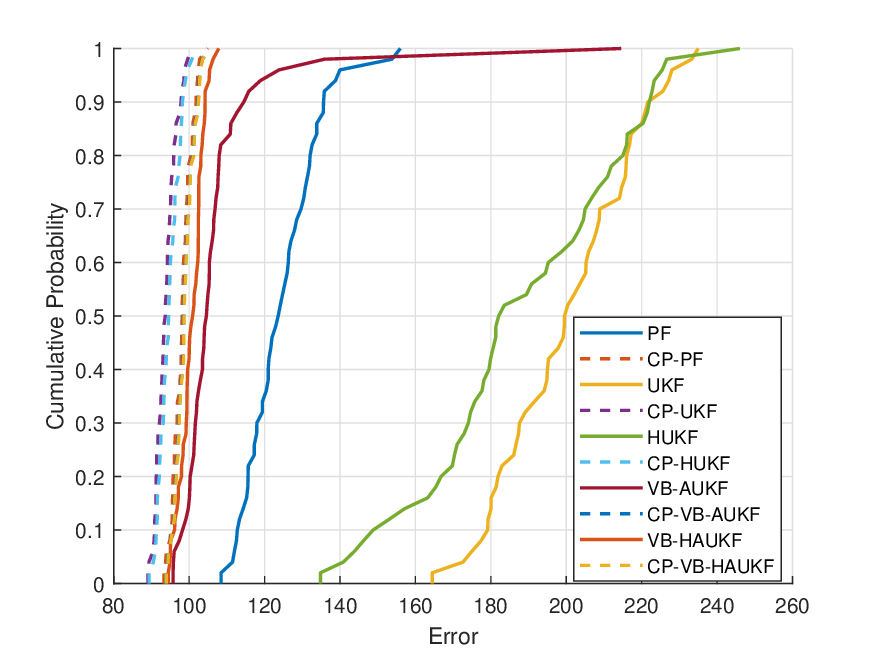}%
    \label{fig6c}%
}
\subfloat[]{%
    \includegraphics[width=0.5\linewidth]{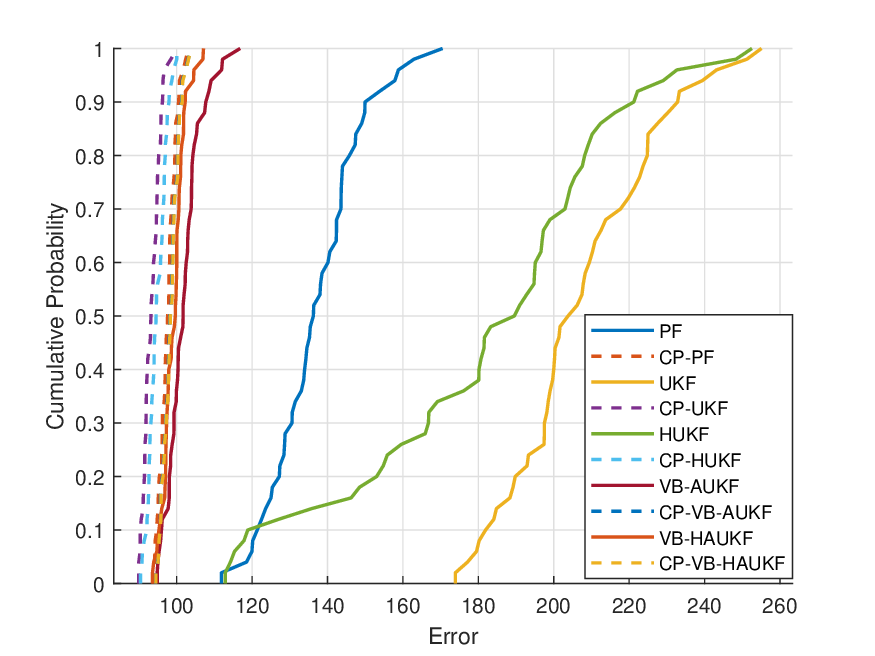}%
    \label{fig6d}%
}
\caption{The CDF curves of different filters for the four different types of measurement noise. (a) Case a. (b) Case b. (c) Case c. (d) Case d.}
\label{fig6}
\end{figure}

\subsection{Implementation Experiment}
In order to evaluate the actual positioning performance of the method proposed in this paper, we designed an indoor positioning and navigation system and conducted experiments on the third and fifth floors of the experimental building at Southwest University. Figure~\ref{fig7} shows the experimental scene and provides the plane profile of the experimental site. The experimental grid was divided every 3 meters, and nodes were marked with location labels. Geomagnetic fingerprints were used to locate the nodes, and the detailed position coordinates were obtained with the assistance of inertial data. 

The basic procedure of the experiment includes the establishment of a fingerprint database, real-time RSSI data acquisition, data fusion prediction via filtering, fingerprint matching of unknown coordinates, and the output of the current position information.

The RSSI measurements were obtained using a high-precision 9-axis geomagnetic sensor in real-time. The experimental equipment consisted of the STM32F103 MCU, WHEELTEC H30 high-precision 9-axis IMU sensor (which includes a 3-axis accelerometer, 3-axis gyroscope, and 3-axis magnetometer), and the related software and hardware tools for geomagnetic matching algorithms, such as the USART HMI, Matlab(R2024a), and XCOM serial debugging tool. Additionally, portable computing devices were employed for real-time data processing and analysis.

The experimental results demonstrate that the accuracy of the indoor single-point fingerprint matching algorithm, when enhanced with the conformal outlier detection method proposed in this paper, improved from 81.25\% to 93.75\%. Furthermore, the positioning error was significantly reduced, from a range of 0.62–6.87 m to just 0.03–0.35 m.
\begin{figure}[H]
\centering
\includegraphics[width=0.45\textwidth]{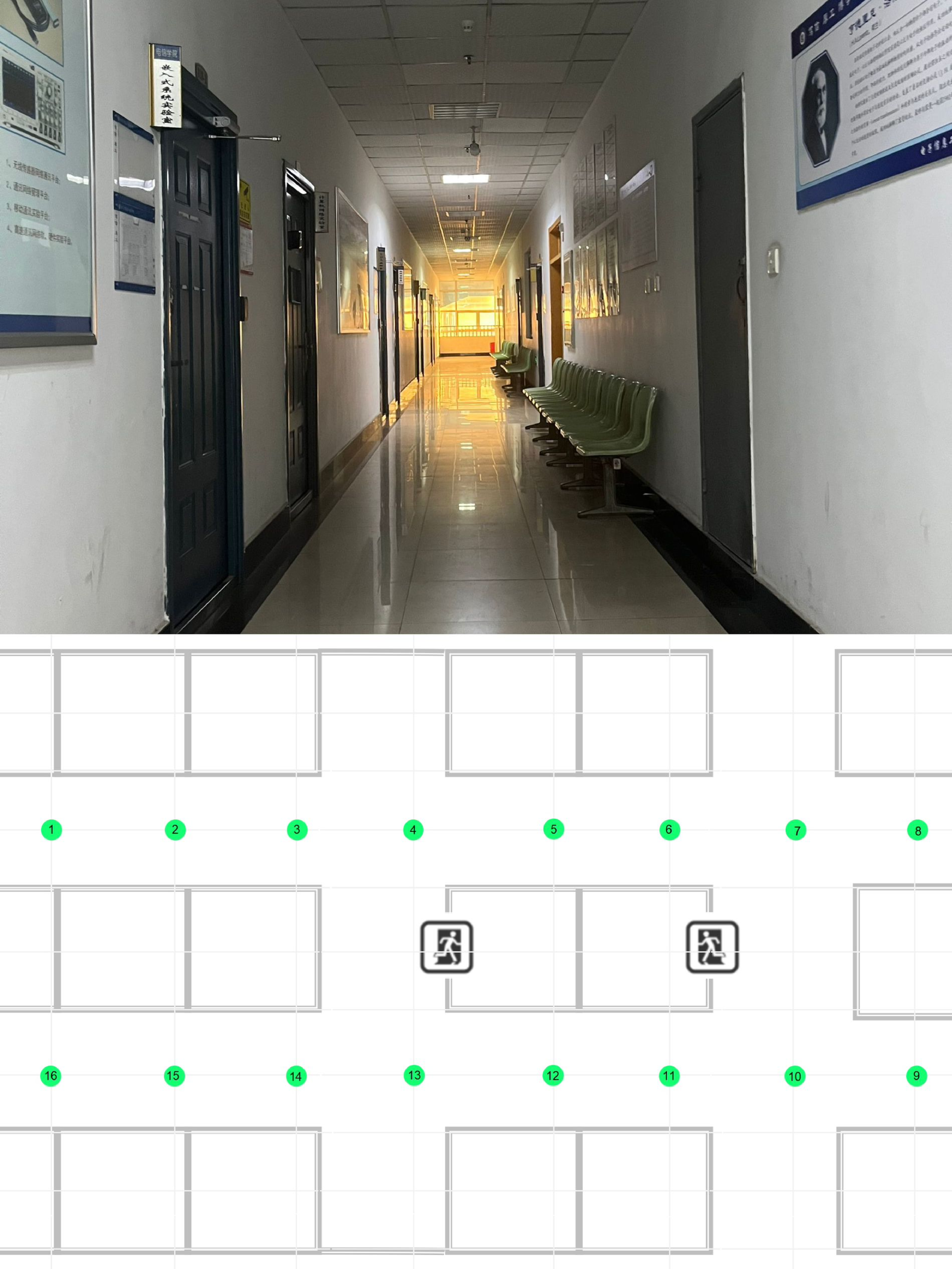}
\caption{Experimental scene.}
\label{fig7}
\end{figure}
\begin{figure}[H]
\centering
\includegraphics[width=0.5\textwidth]{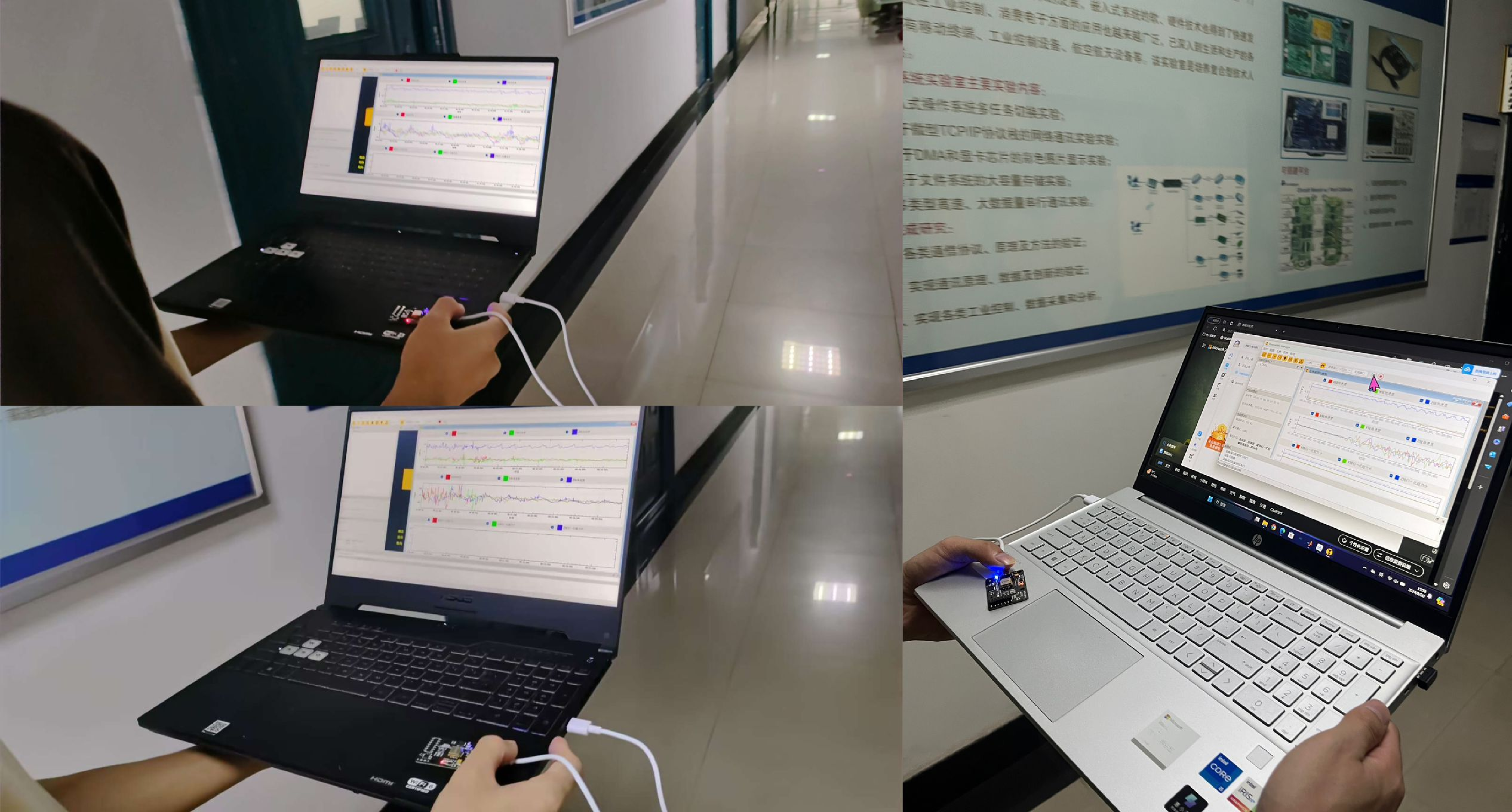}
\caption{The construction of fingerprint database and real-time RSSI acquisition.}
\label{fig8}
\end{figure}
\section{Conclusion}\label{C}
This paper presents a novel approach to improving indoor localization performance by integrating conformal outlier detection with a geomagnetic/inertial-based fingerprinting system. The proposed method addresses common challenges in indoor positioning by providing a robust and adaptive solution that is independent of noise distributions. Through extensive experimental validation, the framework demonstrated a significant improvement in positioning accuracy and reliability, with the fingerprint matching algorithm's accuracy increasing from 81.25\% to 93.75\% and positioning errors reduced from 0.62–6.87 m to 0.03–0.35 m.

The results highlight the effectiveness of conformal prediction in mitigating outlier influence and enhancing the robustness of indoor localization systems. The method's ability to adapt to dynamic environments and non-Gaussian noise makes it particularly well-suited for real-world IoT applications, including autonomous navigation, wearable devices, and smart buildings.


{\appendices
\section*{Proof: Semi-positive Definiteness Proof of Huber-KF Covariance}
\begin{proof}\label{app:huber_proof}
The proof establishes the positive semi-definiteness of the posterior covariance matrix through three steps:

\noindent\textbf{Step 1: Term Decomposition}\\
Decompose the posterior covariance into two terms:
\begin{equation}
P_{k|k} = \underbrace{(I - W_k K_k H)P_{k|k-1}(I - W_k K_k H)^T}_{T_1} + \underbrace{W_k K_k R K_k^T}_{T_2}
\end{equation}

\noindent\textbf{Step 2: Individual Semi-Definiteness Analysis}\\
\begin{itemize}
\item \textit{Term $T_1$:} 
\begin{align*}
T_1 &= A P_{k|k-1} A^T \quad \text{where } A = I - W_k K_k H \\
&\succeq 0 \quad \text{(as $P_{k|k-1} \succ 0$ and $A$ is real)}
\end{align*}

\item \textit{Term $T_2$:}
\begin{align*}
T_2 &= W_k K_k R K_k^T \\
&\succeq 0 \quad \text{(since $R \succ 0$ and $W_k \geq 0$)}
\end{align*}
\end{itemize}

\noindent\textbf{Step 3: Summation Property}\\
The summation preserves semi-definiteness:
\begin{equation}
P_{k|k} = T_1 + T_2 \succeq 0
\end{equation}

\noindent\textbf{Boundary Case Validation}:
\begin{itemize}
\item \textit{Case 1:} $W_k = 0$ (no measurement update):
\begin{equation}
P_{k|k} = P_{k|k-1} \succ 0
\end{equation}

\item \textit{Case 2:} $W_k > 0$:
Both $T_1$ and $T_2$ contribute positively to $P_{k|k} \succeq 0$.
\end{itemize}

\noindent Hence, $P_{k|k} \succeq 0$ holds for all valid $W_k \geq 0$.
\end{proof}

\section*{Proof: Coverage Property of Sliding Window Conformal Prediction}
\label{app:swcp_proof}

Let $\{(X_k, Y_k)\}_{k=1}^K$ be a time series data stream. At each time step $k > w$, we maintain a sliding calibration set $\mathcal{C}_k = \{s_{k-w}, \ldots, s_{k-1}\}$ where $w$ is the window size, and $s_i = s(X_i, Y_i)$ are non-conformity scores.

\subsection*{Modified Exchangeability Assumption}
Assume the scores in the extended window $\mathcal{C}_k \cup \{s_k\}$ are approximately exchangeable within any temporal block of size $w+1$. This weakens the i.i.d. assumption while preserving local stationarity.

\subsection*{Adaptive Quantile Construction}
Define the sliding window quantile estimate:
\[
\hat{q}_k = \inf\left\{ q : \frac{|\{s \in \mathcal{C}_k : s \leq q\}|}{w} \geq \frac{\lceil (w+1)(1-\alpha) \rceil}{w} \right\}
\]

\subsection*{Coverage Analysis}
\begin{theorem}
Under the block exchangeability assumption, the sliding window conformal predictor satisfies:
\[
\left| \mathbb{P}(s_k \leq \hat{q}_k) - (1-\alpha) \right| \leq \frac{C}{\sqrt{w}} 
\]
where $C$ depends on the mixing coefficient of the time series.
\end{theorem}

\begin{proof}
Define the martingale difference sequence:
\[
D_k = \mathbb{I}\{s_k \leq \hat{q}_k\} - (1-\alpha)
\]

The accumulated error satisfies:
\[
\left|\sum_{k=w+1}^K D_k\right| \leq C\sqrt{K\log(1/\delta)} \quad \text{w.p. } 1-\delta
\]

Normalizing by the window size:
\[
\frac{1}{K-w}\left|\sum_{k=w+1}^K D_k\right| \leq \frac{C}{\sqrt{w}} 
\]

Applying the Azuma-Hoeffding inequality for weakly dependent sequences \cite{barber2022predictive}:
\[
\mathbb{P}\left(\left|\frac{1}{w}\sum_{s\in\mathcal{C}_k} \mathbb{I}\{s \leq \hat{q}_k\} - (1-\alpha)\right| \geq \epsilon\right) \leq 2\exp\left(-\frac{w\epsilon^2}{2}\right)
\]

Setting $\epsilon = C/\sqrt{w}$ yields the stated bound.
\end{proof}

\subsection*{Practical Interpretation}
The coverage gap decays as $\mathcal{O}(1/\sqrt{w})$, implying:
\begin{itemize}
\item For $w=1000$, the coverage deviation is bounded by $\sim$3\%
\item Window size controls tradeoff between adaptivity and statistical guarantees
\item The bound holds for $\alpha$-mixing processes \cite{barber2022predictive}
\end{itemize}
}

\bibliographystyle{IEEEtran}
\bibliography{references}

\begin{thebibliography}{10}
\providecommand{\url}[1]{#1}
\csname url@samestyle\endcsname
\providecommand{\newblock}{\relax}
\providecommand{\bibinfo}[2]{#2}
\providecommand{\BIBentrySTDinterwordspacing}{\spaceskip=0pt\relax}
\providecommand{\BIBentryALTinterwordstretchfactor}{4}
\providecommand{\BIBentryALTinterwordspacing}{\spaceskip=\fontdimen2\font plus
\BIBentryALTinterwordstretchfactor\fontdimen3\font minus \fontdimen4\font\relax}
\providecommand{\BIBforeignlanguage}[2]{{%
\expandafter\ifx\csname l@#1\endcsname\relax
\typeout{** WARNING: IEEEtran.bst: No hyphenation pattern has been}%
\typeout{** loaded for the language `#1'. Using the pattern for}%
\typeout{** the default language instead.}%
\else
\language=\csname l@#1\endcsname
\fi
#2}}
\providecommand{\BIBdecl}{\relax}
\BIBdecl

\bibitem{zhu2024enabling}
X.~Zhu, J.~Liu, L.~Lu, T.~Zhang, T.~Qiu, C.~Wang, and Y.~Liu, ``Enabling intelligent connectivity: A survey of secure isac in 6g networks,'' \emph{IEEE Communications Surveys \& Tutorials}, 2024.

\bibitem{xu2019indoor}
Q.~Xu, Y.~Han, B.~Wang, M.~Wu, and K.~R. Liu, ``Indoor events monitoring using channel state information time series,'' \emph{IEEE Internet of Things Journal}, vol.~6, no.~3, pp. 4977--4990, 2019.

\bibitem{urrea2021kalman}
C.~Urrea and R.~Agramonte, ``Kalman filter: historical overview and review of its use in robotics 60 years after its creation,'' \emph{Journal of Sensors}, vol. 2021, no.~1, p. 9674015, 2021.

\bibitem{lin2021variational}
H.~Lin and C.~Hu, ``Variational inference based distributed noise adaptive bayesian filter,'' \emph{Signal Processing}, vol. 178, p. 107775, 2021.

\bibitem{fauss2021minimax}
M.~Fau{\ss}, A.~M. Zoubir, and H.~V. Poor, ``Minimax robust detection: Classic results and recent advances,'' \emph{IEEE Transactions on signal Processing}, vol.~69, pp. 2252--2283, 2021.

\bibitem{erhan2021smart}
L.~Erhan, M.~Ndubuaku, M.~Di~Mauro, W.~Song, M.~Chen, G.~Fortino, O.~Bagdasar, and A.~Liotta, ``Smart anomaly detection in sensor systems: A multi-perspective review,'' \emph{Information Fusion}, vol.~67, pp. 64--79, 2021.

\bibitem{guo2022robust}
G.~Guo, R.~Chen, F.~Ye, Z.~Liu, S.~Xu, L.~Huang, Z.~Li, and L.~Qian, ``A robust integration platform of wi-fi rtt, rss signal, and mems-imu for locating commercial smartphone indoors,'' \emph{IEEE Internet of Things Journal}, vol.~9, no.~17, pp. 16\,322--16\,331, 2022.

\bibitem{pau2021bluetooth}
G.~Pau, F.~Arena, Y.~E. Gebremariam, and I.~You, ``Bluetooth 5.1: An analysis of direction finding capability for high-precision location services,'' \emph{Sensors}, vol.~21, no.~11, p. 3589, 2021.

\bibitem{sambu2022experimental}
P.~Sambu and M.~Won, ``An experimental study on direction finding of bluetooth 5.1: Indoor vs outdoor,'' in \emph{2022 IEEE Wireless Communications and Networking Conference (WCNC)}.\hskip 1em plus 0.5em minus 0.4em\relax IEEE, 2022, pp. 1934--1939.

\bibitem{sun2021indoor}
M.~Sun, Y.~Wang, S.~Xu, H.~Yang, and K.~Zhang, ``Indoor geomagnetic positioning using the enhanced genetic algorithm-based extreme learning machine,'' \emph{IEEE Transactions on Instrumentation and Measurement}, vol.~70, pp. 1--11, 2021.

\bibitem{sun2022indoor}
M.~Sun, Y.~Wang, W.~Joseph, and D.~Plets, ``Indoor localization using mind evolutionary algorithm-based geomagnetic positioning and smartphone imu sensors,'' \emph{IEEE Sensors Journal}, vol.~22, no.~7, pp. 7130--7141, 2022.

\bibitem{qi2023calibration}
M.~Qi, B.~Xue, and W.~Wang, ``Calibration and compensation of anchor positions for uwb indoor localization,'' \emph{IEEE Sensors Journal}, 2023.

\bibitem{ahmad2024recent}
N.~S. Ahmad, ``Recent advances in wsn-based indoor localization: A systematic review of emerging technologies, methods, challenges and trends,'' \emph{IEEE Access}, 2024.

\bibitem{kwon2022adaptive}
B.~Kwon, ``Adaptive fading-memory receding-horizon filters and smoother for linear discrete time-varying systems,'' \emph{Applied Sciences}, vol.~12, no.~13, p. 6692, 2022.

\bibitem{fan2021interacting}
X.~Fan, G.~Wang, J.~Han, and Y.~Wang, ``Interacting multiple model based on maximum correntropy kalman filter,'' \emph{IEEE Transactions on Circuits and Systems II: Express Briefs}, vol.~68, no.~8, pp. 3017--3021, 2021.

\bibitem{chen2023gaussian}
P.-T. Chen, D.~S. Bayard, R.~F. Sharrow, W.~A. Majid, B.~A. Dunst, and J.~L. Speyer, ``A gaussian sum filter for pulsar navigation: Processing single photon arrival time measurements,'' \emph{IEEE Transactions on Control Systems Technology}, vol.~31, no.~6, pp. 2499--2514, 2023.

\bibitem{wang2024improved}
D.~Wang, B.~Wang, H.~Huang, and Y.~Yao, ``An improved robust filter method for sins/dvl system with earth frame error model,'' \emph{IEEE Transactions on Vehicular Technology}, 2024.

\bibitem{li2016variational}
K.~Li, L.~Chang, and B.~Hu, ``A variational bayesian-based unscented kalman filter with both adaptivity and robustness,'' \emph{IEEE Sensors Journal}, vol.~16, no.~18, pp. 6966--6976, 2016.

\bibitem{liu2021variational}
X.~Liu, X.~Liu, Y.~Yang, Y.~Guo, and W.~Zhang, ``Variational bayesian-based robust cubature kalman filter with application on sins/gps integrated navigation system,'' \emph{IEEE Sensors Journal}, vol.~22, no.~1, pp. 489--500, 2021.

\bibitem{davari2021real}
N.~Davari and A.~P. Aguiar, ``Real-time outlier detection applied to a doppler velocity log sensor based on hybrid autoencoder and recurrent neural network,'' \emph{IEEE Journal of Oceanic Engineering}, vol.~46, no.~4, pp. 1288--1301, 2021.

\bibitem{yang2024variational}
B.~Yang, H.~Wang, and Z.~Shi, ``Variational bayesian and generalized maximum-likelihood based adaptive robust nonlinear filtering framework,'' \emph{Signal Processing}, vol. 215, p. 109271, 2024.

\bibitem{waymel2019impact}
Q.~Waymel, S.~Badr, X.~Demondion, A.~Cotten, and T.~Jacques, ``Impact of the rise of artificial intelligence in radiology: what do radiologists think?'' \emph{Diagnostic and interventional imaging}, vol. 100, no.~6, pp. 327--336, 2019.

\bibitem{shafer2008tutorial}
G.~Shafer and V.~Vovk, ``A tutorial on conformal prediction.'' \emph{Journal of Machine Learning Research}, vol.~9, no.~3, 2008.

\bibitem{campos2024conformal}
M.~Campos, A.~Farinhas, C.~Zerva, M.~A. Figueiredo, and A.~F. Martins, ``Conformal prediction for natural language processing: A survey,'' \emph{Transactions of the Association for Computational Linguistics}, vol.~12, pp. 1497--1516, 2024.

\bibitem{xu2023conformal}
C.~Xu and Y.~Xie, ``Conformal prediction for time series,'' \emph{IEEE transactions on pattern analysis and machine intelligence}, vol.~45, no.~10, pp. 11\,575--11\,587, 2023.

\bibitem{strawn2023conformal}
K.~J. Strawn, N.~Ayanian, and L.~Lindemann, ``Conformal predictive safety filter for rl controllers in dynamic environments,'' \emph{IEEE Robotics and Automation Letters}, vol.~8, no.~11, pp. 7833--7840, 2023.

\bibitem{yang2023safe}
S.~Yang, G.~J. Pappas, R.~Mangharam, and L.~Lindemann, ``Safe perception-based control under stochastic sensor uncertainty using conformal prediction,'' in \emph{2023 62nd IEEE Conference on Decision and Control (CDC)}.\hskip 1em plus 0.5em minus 0.4em\relax IEEE, 2023, pp. 6072--6078.

\bibitem{garcia2024multi}
E.~Garcia-Ceja, ``Multi-view conformal learning for heterogeneous sensor fusion,'' \emph{arXiv preprint arXiv:2402.12307}, 2024.

\bibitem{kim2025robust}
D.~Kim, M.~Zecchin, S.~Park, J.~Kang, and O.~Simeone, ``Robust bayesian optimization via localized online conformal prediction,'' \emph{IEEE Transactions on Signal Processing}, 2025.

\bibitem{liu2018geomagnetism}
D.~Liu, S.~Guo, Y.~Yang, Y.~Shi, and M.~Chen, ``Geomagnetism-based indoor navigation by offloading strategy in nb-iot,'' \emph{IEEE Internet of Things Journal}, vol.~6, no.~3, pp. 4074--4084, 2018.

\bibitem{barber2022predictive}
R.~Barber, E.~Candes, A.~Ramdas, and R.~Tibshirani, ``Predictive inference with the jackknife+,'' \emph{Annals of Statistics}, 2022.

\end{thebibliography}

\newpage

 




\vfill

\end{document}